\documentclass{llncs} %
\title{Systematic Classification of Attackers via Bounded Model Checking}
\subtitle{(Extended Version)}
\author{Eric Rothstein-Morris\inst{1} \and Sun Jun\inst{2} \and Sudipta Chattopadhyay\inst{1}}
\institute{$^1$Singapore University of Technology and Design \email{\{eric\_rothstein,sudipta\_chattopadhyay\}@sutd.edu.sg}\\
$^2$ Singapore Management University \\\email{junsun@smu.edu.sg}}

\usepackage{todonotes} 
\usepackage{lscape}
\usepackage{wrapfig}
\usepackage{framed}
\usetikzlibrary{positioning}
\usetikzlibrary{petri}
\usetikzlibrary{shapes,snakes}

\usepackage{tikz}
\usetikzlibrary{arrows}
\usepackage{mathtools}
\usepackage{amssymb} 
\usepackage{todonotes} 
\usepackage{multirow}
\usepackage{mathrsfs}
\usepackage[]{algorithm2e}
\LinesNumbered
\SetKwRepeat{Do}{do}{while}
\SetKw{Break}{break}
\SetKw{Input}{Input:}
\usepackage{soul}
\usepackage{color}
\usepackage{graphicx}
\usepackage{color}
\usepackage{verbatim}
\usepackage{graphicx}
\usepackage{eurosym}
\usepackage{bm}
\usepackage{array}
\usepackage{float}
\usepackage{stmaryrd }
\usepackage{braket}
\usepackage{tikz}
\usepackage{tikz-cd}
\usepackage{mathtools}
\usepackage{listings}
\usepackage{mdframed}

\newcommand{\Bool}{{\mathbb{B}}}
\newcommand{\Always}{{\Globally}}
\newcommand{\Globally}{{\square}}
\newcommand{\vect}[1]{\ensuremath{\overrightarrow{\bm{#1}}}}
\newcommand{\AsSequence}[1]{{{\left<#1\right>}}}
\newcommand{\ThePowersetOf}[1]{{\Powerset\!\left({#1}\right)}}
\newcommand{\Powerset}{{\mathscr{P}}}
\begin{document}
\maketitle
\begin{abstract}
In this work, we study the problem of verification of systems in the presence of attackers using bounded model checking. Given a system and a set of security requirements, we present a methodology to generate and classify attackers, mapping them to the set of requirements that they can break. A naive approach suffers from the same shortcomings of any large model checking problem, i.e., memory shortage and exponential time. To cope with these shortcomings, we describe two sound heuristics based on cone-of-influence reduction and on learning, which we demonstrate empirically by applying our methodology to a set of hardware benchmark systems.
\end{abstract}

\section{Introduction}
\label{sec:Introduction}
\subsubsection{Problem Context.} 
Some systems are designed to provide security guarantees in the presence of attackers. For example, the Diffie-Hellman key agreement protocol guarantees perfect forward secrecy \cite{Gunther1990,Menezes1996} (PFS), i.e., that the session key remains secret even if the long-term keys are compromised. These security guarantees are only valid in the context of the attacker models for which they were proven; more precisely, those guarantees only hold for attackers that fit those or weaker attacker models. For instance, PFS describes an attacker model (i.e., an attacker that can compromise the long-term keys, \emph{and only those}), and a property that is guaranteed in the presence of an attacker that fits the model (i.e., confidentiality of the session keys). However, if we consider an attacker model that is stronger (e.g., an attacker that can directly compromise the session key), then Diffie-Hellman can no longer guarantee the confidentiality of the session keys. Clearly, it is difficult to provide any guarantees against an attacker model that is too capable, so it is in the interest of the system designer to choose an adequate attacker model that puts the security guarantees of the system in the context of realistic and relevant attackers.

Consider the following research question: \textbf{RQ1)} given a system and a list of security requirements, how do we systematically generate attackers that can potentially break these requirements, and how do we verify if they are successful? We approach this question at a high level for a system $S$ with a set $C$ of $n$ components, and a set of security requirements $R$ as follows. Let $A$ be a subset of $C$; the set $A$ models an attacker that can interact with $S$ by means of each component $c$ in it. More precisely, for every component $c$ in $A$, the attacker can change the value of $c$ at any time and any number of times during execution, possibly following an attack strategy. Considering the exponential size of the set of attackers (i.e., $2^{n}$), a brute-force approach to checking whether each of those attackers breaks each requirement in $R$ is inefficient for two reasons: 1) an attacker $A$ may only affect an isolated part of the system, so requirements that refer to other parts of the system should not be affected by the presence of $A$, and 2) if some attacker $B$ affects the system in a similar way to $A$ (e.g., if they control a similar set of components), then the knowledge we obtain while verifying the system in the presence of $A$ may be useful when verifying the system in the presence of $B$. These two reasons motivate a second research question: \textbf{RQ2)} which techniques can help us efficiently \emph{classify} attackers, i.e., to map each attacker to the set of requirements that it breaks? 

To answer these two research questions in a more concrete and practical context, we study systems modelled by \emph{And-inverter Graphs (AIGs)} (see  \cite{AIGs,AIGs2}). AIGs describe hardware models at the bit-level \cite{AIGER}, and have attracted the attention of industry partners including IBM and Intel \cite{HWMCC2014BM}. Due to being systems described at bit-level, AIGs present a convenient system model to study the problem of attacker classification, because the range of actions that attackers have over components is greatly restricted: either the attacker leaves the value of the component as it is, or the attacker negates its current value. However, this approach can be generalised to other systems by considering non-binary ranges for components, and by allowing attackers to choose any value in those ranges.

\subsubsection{Contributions.} In this paper, we provide:
\begin{itemize}
\item a formalisation of attackers of AIGs and how they interact with systems, 
\item a methodology to perform bounded model checking while considering the presence of attackers,
\item a set of heuristics that efficiently characterise attacker frontiers for invariant properties using bounded model checking,
\item experimental evidence of the effectiveness of the proposed methodology and heuristics.
\end{itemize}
\section{Preliminaries}
\label{sec:preliminaries}
In this section, we provide the foundation necessary to formally present 
the problem of model checking {And-inverter Graph} (AIGs) in the presence of attackers. 
Let $\Bool=\set{0,1}$ be the set of booleans. An \emph{And-inverter Graph} 
models a system of equations that has $m$ boolean 
inputs, $n$ boolean state variables and $o$ boolean gates. The elements in the set 
$W=\set{w_1, \ldots, w_m}$ represent the \emph{inputs}, the elements in 
$V=\set{v_1, \ldots, v_n}$ represent the \emph{latches}, and the elements in 
$G=\set{g_1, \ldots, g_o}$ represent the \emph{and-gates}. We assume that $W$, $V$ and $G$ are pairwise disjoint, and we define the set of \emph{components} $C$ by $C\triangleq W\cup V \cup G$. 

An \emph{expression} $e$ is described by the grammar $e::= 0\ |\ 1\ |\ c\ |\ \lnot c$, where $c \in C$. The set of all expressions is $E$. We use discrete time steps $t=0,1,..$ to describe the system of equations. To each latch $v\in V$ we associate a transition {equation} of the form $v(t+1) = e(t)$ and an initial equation of the form $v(0)=b$, where $e\in E$ and $b\in \Bool$. To each gate $g\in G$ we associate an equation of the form $g(t)=e_1(t)\land e_2(t)$, where $e_1,e_2\in E$. 

\begin{quote}
\begin{example}
\label{ex:simple}
Figure~\ref{fig:Example} shows an example AIG with $W=\set{w_1,w_2}$, $V=\set{v_1}$ and $G=\set{g_1,g_2}$. The corresponding system of equations is
\begin{align*}
&v_1(0)=1, & v_1(t+1)=\lnot g_2(t),\\
&g_1(t)=\lnot w_1(t) \land \lnot w_2(t), & g_2(t)=g_1(t) \land \lnot v_1(t).
\end{align*}

\end{example}
\end{quote}
\begin{figure}[!t]
\begin{minipage}{0.6\textwidth}
\begin{framed}
\includegraphics[width=\textwidth]{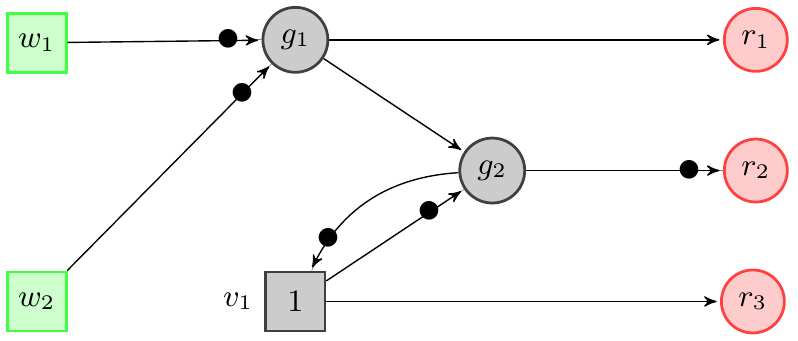}
\end{framed}
\end{minipage}
\begin{minipage}{0.35\textwidth}
\centering
\begin{framed}
\includegraphics[width=\textwidth]{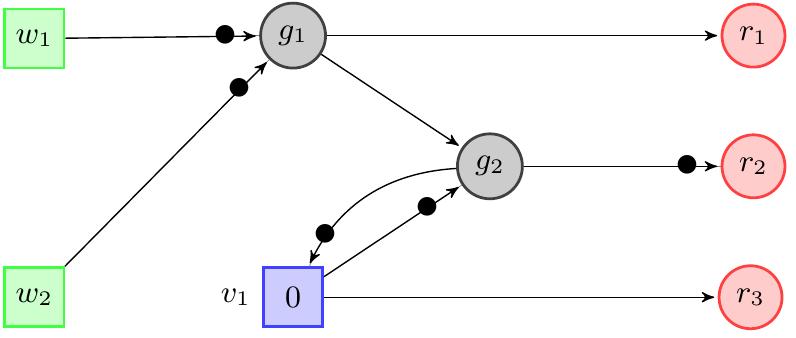}\\
\end{framed}
\begin{framed}
\includegraphics[width=\textwidth]{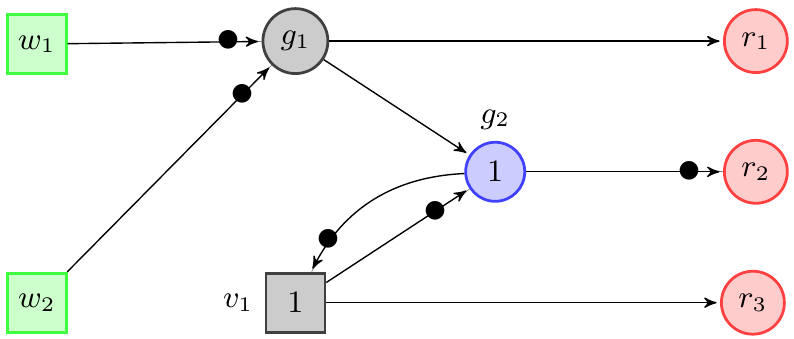}
\end{framed}
\end{minipage}
\caption{\textbf{Left:} And-inverter graph describing a system with two
inputs $w_1$ and $w_2$ (green boxes), one latch $v_1$ with initial value 1 (grey box), two gates $g1$ and $g2$ (gray circles), and three invariant requirements $r_1=\Always g_1$, $r_2=\Always \lnot g_2$ and $r_3=\Always v_1$ (red circles). Arrows represent logical dependencies, and bullets in the arrows imply negation. \newline
\textbf{Right above}: an attacker that controls latch $v_1$ can set its value to 0 and break $r_2$ and $r_3$ in 0 steps.\newline
\textbf{Right below}: an attacker that controls gate $g_2$ can set its value to 1 and break $r_2$ in 0 steps and $r_3$ in 1 step, because the value of $v_1$ at time 1 is 0. }
\label{fig:Example}
\end{figure}

The states of a system are all the different valuations of the variables in $V$; 
formally, a state $\vect{v} \colon V \rightarrow \Bool$ is a map of elements of $V$ to booleans (i.e. a vector of bits). Similarly, the 
valuations of the variables in $W$ are all the inputs to the systems; 
again, an input $\vect{w}\colon W \rightarrow \Bool$ is a map of elements of $W$ to booleans (also a vector of bits). We refer to the set of all states by $\vect{V}
$, and to the set of all inputs by $\vect{W}
$.
For  $t=0,1,...$, we denote the state of the system at time $t$ by $\vect{v}(t)$, with 
$\vect{v}(t)\triangleq \AsSequence{v_1(t), \ldots, v_n(t)}$. 
The initial state is $\vect{v}(0)$, defined by the initial equations for the latches. Similarly, we denote the input of the system at time $t$ by $\vect{w}(t)$, with $\vect{w}(t)\triangleq\AsSequence{w_1(t), \ldots, w_n(t)}$. There are no restrictions or assumptions over $\vect{w}(t)$, so it can take any value in $\vect{W}$.
\begin{quote}
In Example~\ref{ex:simple}, the set of states is $\vect{V}=\set{\AsSequence{0},\AsSequence{1}}$, the set of all inputs is $\vect{W}=\set{\AsSequence{0,0},\AsSequence{0,1},\AsSequence{1,0},\AsSequence{1,1}}$, and the initial state is $\vect{v}(0)=\AsSequence{1}$.
\end{quote}

Given an expression $e\in E$, the invariant $\Always e$ is the property that requires $e(t)$ to be true for all $t\geq 0$. The system $S$ \emph{fails} the invariant $\Always e$ iff there exists a finite sequence of input vectors $\AsSequence{\vect{w}_0, \ldots, \vect{w}_t}$ such that, if we assume $\vect{w}(t)=\vect{w}_t$, then $e(t)$ is false. The system \emph{satisfies} the invariant $\Always e$ if no such sequence of input vectors exists. Every expression $e$ represents a boolean predicate over the state of the latches of the system, and can be used to characterise states that are (un)safe. These expressions are particularly useful in safety-critical hardware, as they can signal the approach of a critical state.

\begin{quote}
In Example~\ref{ex:simple}, we define three requirements: $r_1\triangleq\Always g_1$, $r_2 \triangleq\Always \lnot g_2$, and $r_3\triangleq\Always v_1$. This system satisfies $r_2$ and $r_3$, but it fails $r_1$ because $w_1=1$ and $w_2=0$ results in $g_1(0)$ being 0.
\end{quote}

The \emph{Cone-of-Influence} (COI) is a mapping from an expression to the components that can potentially influence its value. We obtain the COI of an expression $e \in E$, denoted $\blacktriangledown(e)$, by transitively tracing its dependencies to inputs, latches and gates. More precisely, 
\begin{itemize}
\item $\blacktriangledown(0)=\emptyset$ and $\blacktriangledown(1)=\emptyset$;
\item if $e=\lnot c$ for $c \in C$, then $\blacktriangledown(e)= \blacktriangledown(c)$;
\item if $e=w$ and $w$ is an input, then $\blacktriangledown(e)=\set{w}$;
\item if $e=v$ and $v$ is a latch whose transition equation is $l(t+1) = e'(t)$, then $\blacktriangledown(e)=\set{v} \cup \blacktriangledown(e')$;
\item if $e=g$ and $g$ is a gate whose equation is $g(t) = e_1(t) \land e_2(t)$ then $\blacktriangledown(e)=\set{g} \cup \blacktriangledown(e_1)\cup \blacktriangledown(e_2)$.
\end{itemize}
The COI of a requirement $r=\Always e$ is $\blacktriangledown(r)\triangleq \blacktriangledown(e)$.

\begin{quote}
In Example~\ref{ex:simple}, the COI for the requirements are $\blacktriangledown(\Always g_1)=\set{g_1,w_1,w_2}$, and $\blacktriangledown(\Always g_2)=\blacktriangledown(\Always g_3)=\set{g_1,g_2,v_1,w_1,w_2}$.
\end{quote}

The set of \emph{sources} of an expression $e\in E$, denoted $\mathtt{src}(e)$ is the set of latches and inputs in the COI of $e$; formally, $\mathtt{src}(e)\triangleq \blacktriangledown(e) \cap (V \cup W)$. The \emph{Jaccard index} of two expressions $e_1$ and $e_2$ is equal to $
\frac{|\mathtt{src}(e_1)\cap \mathtt{src}(e_2)|}{|\mathtt{src}(e_1)\cup\mathtt{src}(e_2)|}
$. This index provides a measure of how similar the sources of $e_1$ and $e_2$ are. 

The \emph{dual cone-of-influence} (IOC) of a component $c \in C$, denoted $\blacktriangle(c)$, is the set of components influenced by $c$; more precisely 
$\blacktriangle(c) \triangleq \set{c' \in  C | c\in \blacktriangledown(c')}.$

\section{Motivational Example}
\label{sec:Example}
In this section, we provide a motivational example of the problem of model checking compromised systems, and we illustrate how to classify attackers given a list of security requirements. Consider a scenario where an attacker $A$ controls the gate $g_2$ of Example~\ref{ex:simple}. By controlling $g_2$, we mean that $A$ can set the value of $g_2(t)$ at will for all $t\geq 0$. Since $r_2=\Always \lnot g_2$, it is possible for $A$ to break $r_2$ by setting $g_2(0)$ to $1$. We note that the same strategy works to break both requirements, but it need not be in the general case; i.e., an attacker may have one strategy to break one requirement, and a different strategy to break another. $A$ can also break $r_3=\Always v_1$, because, if $A$ sets $g_2(0)$ to 1, then $v_1(1)$ is equal to $0$. Since the original system fails to enforce $r_1$, we say that $A$ has the \emph{power to break the requirements} $r_1$, $r_2$ and $r_3$. Now, consider a different attacker $B$ which only controls the gate $g_1$. No matter what value $B$ chooses for $g_1(t)$ for all $t$, it is impossible for $B$ to break $r_2$ or $r_3$, so we say that $B$ only has the power to break $r_1$. 

If we allow attackers to control any number of components, then there are $8$ different attackers, described by the subsets of $\set{v_1,g_1,g_2}$. We do not consider attackers that control inputs, because the model checking of invariant properties requires the property to hold for all inputs, so giving control of inputs to an attacker does not make it more powerful (i.e. the attacker cannot break more requirements than it already could without the inputs). Figure~\ref{fig:borders} illustrates the classification of attackers depending on whether they can break a given requirement or not. Based on it, we can provide the following security guarantees: 1) the system cannot enforce $r_1$, and 2) that the system can only enforce $r_2$ and $r_3$ in the presence of attackers that are as capable to interact with the system as $\set{g_1}$ (i.e. they only control $g_1$ or nothing).

According to the classification, attacker $\set{g_2}$ is as powerful as the attacker $\set{v_1,g_1,g_2}$, since both attackers can break the same requirements $r_1$, $r_2$ and $r_3$. This information may be useful to the designer of the system, because it may prioritise attackers that control less components but are as powerful as attackers that control more when deploying defensive mechanisms.

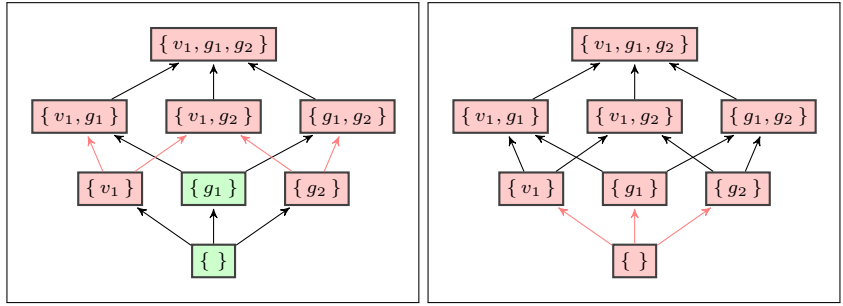
\begin{figure}[!t]
\centering
\begin{minipage}{0.45\textwidth}
{
\scriptsize
\begin{framed}
\begin{tikzpicture}[node distance=1cm,>=stealth',auto, pin distance=0.1cm]
\tikzstyle{safe}=[rectangle,thick,draw=black!75,fill=green!20,minimum size=2mm]
\tikzstyle{fail}=[rectangle,thick,draw=black!75,fill=red!20,minimum size=2mm]
\tikzstyle{dnc}=[rectangle,thick,draw=black!75,fill=red!20,minimum size=2mm]
\node [dnc] (all) {$\set{v_1,g_1,g_2}$};
\node [dnc] (v1g2) [below= 0.5cm of all] {$\set{v_1,g_2}$}
edge [post](all);
\node [dnc] (v1g1) [left = 0.5cm of v1g2] {$\set{v_1,g_1}$}
edge [post](all);
\node [dnc] (g1g2) [right = 0.5cm of v1g2] {$\set{g_1,g_2}$}
edge [post](all);
\node [safe] (g1) [below= 0.5cm of v1g2] {$\set{g_1}$}
edge [post](v1g1)
edge [post](g1g2);
\node [fail] (v1) [left = 0.5cm of g1] {$\set{v_1}$}
edge [post,color=red!50](v1g1)
edge [post,color=red!50](v1g2);
\node [fail] (g2) [right = 0.5cm of g1] {$\set{g_2}$}
edge [post,color=red!50](v1g2)
edge [post,color=red!50](g1g2);
\node [safe] (empty) [below= 0.5cm of g1] {$\set{}$}
edge [post](v1)
edge [post](g1)
edge [post](g2);
 \end{tikzpicture}
 \end{framed}}

\end{minipage}
\begin{minipage}{0.45\textwidth}
{\scriptsize
\begin{framed}
\begin{tikzpicture}[node distance=1cm,>=stealth',auto, pin distance=0.1cm]
\tikzstyle{safe}=[rectangle,thick,draw=black!75,fill=green!20,minimum size=2mm]
\tikzstyle{fail}=[rectangle,thick,draw=black!75,fill=red!20,minimum size=2mm]
\tikzstyle{fail2}=[diamond, draw=black!75,fill=red!20,minimum size=1mm]
\tikzstyle{dnc}=[rectangle,thick,draw=black!75,fill=red!20,minimum size=2mm]
\node [dnc] (all) {$\set{v_1,g_1,g_2}$};
\node [dnc] (v1g2) [below= 0.5cm of all] {$\set{v_1,g_2}$}
edge [post](all);
\node [dnc] (v1g1) [left = 0.5cm of v1g2] {$\set{v_1,g_1}$}
edge [post](all);
\node [dnc] (g1g2) [right = 0.5cm of v1g2] {$\set{g_1,g_2}$}
edge [post](all);
\node [dnc] (g1) [below= 0.5cm of v1g2] {$\set{g_1}$}
edge [post](v1g1)
edge [post](g1g2);
\node [dnc] (v1) [left = 0.5cm of g1] {$\set{v_1}$}
edge [post](v1g1)
edge [post](v1g2);
\node [dnc] (g2) [right = 0.5cm of g1] {$\set{g_2}$}
edge [post](v1g2)
edge [post](g1g2);
\node [fail] (empty) [below= 0.5cm of g1] {$\set{}$}
edge [post,color=red!50](v1)
edge [post,color=red!50](g1)
edge [post,color=red!50](g2);
 \end{tikzpicture}
  \end{framed}}
\end{minipage}
\caption{Left: classification of attackers for requirements $r_2$ and $r_3$. Right: classification of attackers for requirement $r_1$. A green attacker cannot break the requirement, while a red attacker can.}
\label{fig:borders}
\end{figure}

\section{Bounded Model Checking of Compromised Systems}
\label{sec:bmc}
We recall the research questions that motivate this work: \textbf{RQ1)} given a system and a list of security requirements, how do we systematically generate attackers that can potentially break these requirements, and how do we verify if they are successful? and \textbf{RQ2)} which techniques can help us efficiently \emph{classify} attackers, i.e., to map each attacker to the set of requirements that it breaks? 
In this section, we aim to answer these research questions on a theoretical level by formalising the problem of attacker classification via bounded model checking AIGs in the presence of attackers. More precisely, to answer \textbf{RQ1}, we formalise attackers and their interactions with systems, and we show how to systematically generate bounded model checking problems that solve whether some given attackers can break some given requirements. We then propose two methods for the classification of attackers: 1) a brute-force method that creates a model checking problem for each attacker-requirement pair, and 2) a method that incrementally empowers attackers to find ``minimal attackers,'' since minimal attackers represent large portions of the universe of attackers thanks to a monotonicity relation between the set of components controlled by the attacker and the set of requirements that the attacker can break. The latter method is a theoretical approach to answer \textbf{RQ2}, while its practical usefulness is evaluated in Section~\ref{sec:evaluation}.

\subsection{Attackers and Compromised Systems}
Since an AIG describes a system of equations, to incorporate the actions of an attacker $A$ into the system, we modify the equations that are associated to the components controlled by $A$. Let $S=(W,V,G)$ be a system described by an AIG, let $R=\set{r_1, \ldots, r_n}$ be a set of invariant requirements for $S$, and let $C=W\cup V\cup G$ be the set of components of $S$. By definition, an \emph{attacker} $A$ is any subset of $C$. If a component $c$ belongs to an attacker $A$, then $A$ has the \emph{capability to interact with $S$ through} $c$. We modify the equations of every latch $v\in V$ to be parametrised by an attacker $A$ as follows: the original transition equation $v(t+1)=e(t)$ and the initial equation $v(0)=b$ changes to
\begin{align}
\label{eq:badLatch}
v(t+1) = \begin{cases}
e(t), \quad &\text{if $v\not \in A$;}\\
A_v(t+1), \quad &\text{otherwise},
\end{cases}
\quad 
v(0)= \begin{cases}
b, \quad &\text{if $v \not \in A$;}\\
A_v(0), \quad &\text{otherwise.}
\end{cases}
\end{align}
where $A_v(t)$ is a value chosen by the attacker $A$ at time $t$. Similarly, we modify the equation of gate $g\in G$ as follows: the original equation $g(t)=e_1(t)\land e_2(t)$ changes to
\begin{align}
\label{eq:badGate}
g(t) = \begin{cases}
e_1(t)\land e_2(t), \quad &\text{if $g\not \in A$;}\\
A_g(t), \quad &\text{otherwise},
\end{cases}
\end{align}
where $A_g(t)$ is, again, a value chosen by the attacker $A$ at time $t$.

An \emph{attack} $\vect{a}\colon A\rightarrow \Bool$ is a map of components in $A$ to booleans. An \emph{attack strategy} is a finite sequence of attacks $(\vect{a}_0, \vect{a}_1, \ldots, \vect{a}_t)$ that fixes the values of all $A_c(k)$ (used in the equations above) by $A_c(k)=\vect{a}_k(c)$, with $c\in A$ and $0 \leq k \leq t$.

We use $A[S]$ to interpret the system $S$ under the influence of attacker $A$ (i.e.,  $A[S]$ is the modified system of equations). 
Given a requirement $r\in R$ with $r=\Always e$, we say that $A$ \emph{breaks the requirement $r$} if and only if there exists a sequence of inputs of length $k$ and an attack strategy of length $k$ such that $e(k)$ is false. We denote the set of requirements that $A$ breaks by  $A[R]$. 

Finally, we define two partial orders for attackers: \textbf{1)} an attacker $A_i$ is strictly less \emph{capable} (to interact with the system) than an attacker $A_j$ in the 
context of $S$ 
iff $A_i\subseteq A_j$ and $A_i \neq A_j$. The {attacker $A_i$ is equally capable to attacker $A_j$ iff $A_i= A_j$}; and \textbf{2)} an attacker $A_i$ is strictly less \emph{powerful} than an attacker $A_j$ in the 
context of $S$ and $R$ 
iff $A_i[R]\subseteq A_j[R]$ and $A_i[R]\neq A_j[R]$. Similarly, {attacker $A_i$ is equally powerful to attacker $A_j$ iff $A_i[R]= A_j[R]$}. We simply state that $A_i$ is less capable than $A_j$ if $S$ is clear from the context. Similarly, we simply say that $A_i$ is less powerful than $A_j$ if $S$ and $R$ are clear from the context.

We can now properly present the problem of \emph{attacker classification}.
\begin{definition} [Attacker Classification via Model Checking]
\label{def:AttackerQuantification}
Given a system $S$, a set of requirements $R$, and a set of $h$ attackers 
$\set{A_1, \ldots , A_h}$, for every attacker $A$, we compute the set $A[R]$ of requirements that $A$ can break by performing model checking of each requirement in $R$ on the compromised system $A[S]$. 
\end{definition}

Definition~\ref{def:AttackerQuantification} assumes that exhaustive model checking is possible for $S$ and the compromised versions $A[S]$ for all attackers $A$. However, if exhaustive model checking is not possible (e.g., due to time limitations or memory restrictions), we consider an alternative formulation for \emph{Bounded Model Checking} (BMC): 
\begin{definition} [Attacker Classification via Bounded Model Checking]
\label{def:BoundedModelCheckingOfSystems}
Let  $S$ be a system, $R$ be a set of requirements, and $t$ be a natural number. Given a set of attackers $\set{A_1, \ldots , A_h}$, for each attacker $A$ , we compute the set $A[R]$ of requirements that $A$ can break \emph{using a strategy of length up to $t$} on the compromised system $A[S]$. 
\end{definition}

In the following, we show how to construct a SAT formula that describes the attacker classification problem via bounded model checking. 

\subsection{A SAT Formula for BMC up to $t$ Steps}
 For a requirement $r=\Always e$ and a time step $t\geq 0$, we are interested in finding an assignment of sources and attacker actions (i.e., an attack strategy) such that, for $0\leq k \leq t$, the value of $e(k)$ is false. We define the proposition $\mathtt{goal}(r,t)$ by
\begin{align}
\mathtt{goal}(\Always e,t) \triangleq \bigvee_{k=0}^t{\lnot e(k)}, 
\end{align}
 
We must inform the SAT solver of the equalities and dependencies between expressions given by the definition of the AIG (e.g., that $e(k) \Leftrightarrow \lnot v_1(k)$). Inspired by the work of Biere \emph{et al.} \cite{BMCWithoutBDDs}, we transform the equations into a Conjunctive Normal Form formula (CNF) that the SAT solver can work with. To transform Equations \ref{eq:badLatch} and \ref{eq:badGate}, we use \emph{Tseitin encoding} as follows:  the equation
\begin{align*}
v(0)= \begin{cases}
b, \quad &\text{if $v \not \in A$;}\\
A_v(0), \quad &\text{otherwise}
\end{cases}
\text{, becomes
{\small$\left(\lnot v^{\downarrow} \lor (v(0) \Leftrightarrow b ) \right)\land \left(v^{\downarrow} \lor (v(0) \Leftrightarrow A_v(0)) \right)$}}
\end{align*}
where $v^{\downarrow}$ is a literal that marks whether the latch $v$ is an element of the attacker $A$ currently being checked, i.e., we assume that $v^{\downarrow}$ is true if $v\in A$, and we assume that $v^{\downarrow}$ is false if $v\not \in A$. Consequently, if $v\not \in A$, then $v(0) \Leftrightarrow b$ must be true, and if $v \in A$, then $v(0)  \Leftrightarrow A_v(0) $ must be true. We denote this new proposition by $\mathtt{encode}(v,0)$, and it characterises the initial state.

Similarly, for $0\leq k<t$, each equation of the form 
\begin{align*}
&v(k+1) = \begin{cases}
e(k), \quad &\text{if $v\not \in A$;}\\
A_v(k+1), \quad &\text{otherwise},
\end{cases}\quad
\\\text{becomes }
&\left(v^{\downarrow} \lor (v(k+1) \Leftrightarrow e(k) ) \right)\land \left(\lnot v^{\downarrow} \lor (v(k+1) \Leftrightarrow A_v(k+1)) \right).
\end{align*}
We denote these new propositions by $\mathtt{encode}(v,k)$. We now use the Tseitin encoding of $p \Leftrightarrow (q \land r)$, i.e., $(p \lor \lnot q \lor \lnot r)\land (\lnot p \lor q)\land  (\lnot p \lor r)$, to encode gates. For $0\leq k \leq t$, the equation 
\begin{align*}
&g(k) = \begin{cases}
e_1(k)\land e_2(k), \quad &\text{if $g\not \in A$;}\\
A_g(k), \quad &\text{otherwise},
\end{cases}\\\text{becomes }
&\left(g^{\downarrow} \lor (g(k) \Leftrightarrow e_1(k)\land e_2(k) ) \right)\land \left(\lnot g^{\downarrow} \lor (g(k) \Leftrightarrow A_g(k)) \right),
\end{align*}
where $g^{\downarrow}$ is a literal that marks whether the gate $g$ is an element of the attacker $A$ currently being checked in a similar way that the literal $v^\downarrow$ works for the latch $v$. We denote this new proposition by $\mathtt{encode}(g,k)$.

Consider a component $c$ and an attacker $A$. If $c\in A$, then we assume $c^{\downarrow}$ and the value of component $c$ at time $k$ depends on the source literal $A_c(k)$; if $c\not\in A$, then we assume $\lnot c^{\downarrow}$ so we use the original semantics of the system (which depends only on the input literals). During bounded model checking, we need to find an assignment of inputs in $W$ and attacker actions for each component $c$ in $A$ over $t$ steps; thus, we need to assign at least ${|W\times A|\times t}$ literals.

The SAT problem for checking whether requirement $r$ is safe up to $t$ steps, denoted $\mathtt{check}(r,t)$, is defined by 
\begin{align}
\label{eq:naiveCheck}
\mathtt{check}(r,t)\triangleq\mathtt{goal}(r,t)\land\! \bigwedge_{c\in (V \cup G)}\left( \bigwedge_{k=0}^{t}{\mathtt{encode}(c,k)}\right).
\end{align}
\begin{proposition}
\label{prop:Correctness}
For a given attacker $A$ and a requirement $r=\Always e$, if we assume the literal $c^{\downarrow}$ for all $c \in A$ and we assume $\lnot x^{\downarrow}$ for all $x\not\in A$ (i.e., $x\in (V \cup G)-A$), then $A$ can break the requirement $r$ in $t$ steps (or less) if and only if $\mathtt{check}(r,t)$ is satisfiable.
\end{proposition}
\begin{proof}

We first show that if $\mathtt{check}(r,t)$ is satisfiable, then $A$ can break the requirement. If $\mathtt{check}(r,t)$ is satisfiable then $\mathtt{goal}(r,t)$ is satisfiable, and $e(k)$ is false for some $k\leq t$. Moreover, there is an assignment of inputs $\vect{w}(k)$ and attacker actions $A_c(k)$, such that the $\mathtt{encode(c,k)}$ propositions are satisfied. By taking $\vect{w}_k =\vect{w}(k)$ and $\vect{a}_k(c)=A_c(k)$, we provide a witness input vector sequence and a witness attack strategy which proves that $A$ can break $r$ in $k$ steps (i.e., in $t$ steps or less since $k\leq t$).

We now show that if $A$ can break the requirement in $t$ steps or less, then $\mathtt{check}(r,t)$ is satisfiable. Since $A$ breaks $r=\Always e$ in $t$ steps or less, there must be an assignment of the input vector sequence $(\vect{w}_0, \ldots, \vect{w}_k)$ and an attacker strategy $(\vect{a}_0, \ldots, \vect{a}_k)$ such that $e(k)$ is false for some $k\leq t$; this satisifes $\mathtt{goal}(r,t)$, which, in turn, makes $\mathtt{check}(r,t)$ satisfiable.

\qed
\end{proof}

Algorithm~\ref{alg:BadQuantification} describes a naive strategy to compute the sets $A[R]$ for each attacker $A$; i.e. the set of requirements that $A$ breaks in $t$ steps (or less). Algorithm~\ref{alg:BadQuantification} works by solving, for each of the $2^{|V \cup G|}$ different attackers, a set of $|R|$ SAT problems, 
{
each of which has a size of at least $\mathcal{O}\left({|C|\times t}\right)$ on the worst case.}

\begin{figure}[!t]
\centering
{
\begin{framed}
\begin{algorithm}[H]
 \KwData{system $S=(W,V,G)$, a time step $t\geq 0$, a set of requirements $R$.}
 \KwResult{A map that maps the attacker $A$ to $A[R]$.}
Map $\mathcal{H}$\;
\For{\!\!\textbf{each} $r \in R$}
{
\For{\!\!\textbf{each} $A$ such that $A \subseteq (V\cup G)$}
	{
		\If{$\mathtt{check}(r,t)$ is satisfiable while assuming $c^\downarrow$ for all $c\in A$}
		{
			insert $r$ in $\mathcal{H}(A)$\;
		}
	 }
}
 \Return $\mathcal{H}$\;
 \caption{Naive attacker classification algorithm.}
 \label{alg:BadQuantification}
\end{algorithm}
\end{framed}
}
\end{figure}

In the rest of the section, we propose two sound heuristics in an attempt to improve Algorithm~\ref{alg:BadQuantification}: the first technique aims to reduce the size of the SAT formula, while the other aims to record and propagate the results of verifications among the set of attackers so that some calls to the SAT solver can be avoided.

\subsection{Isolation and Monotonicity}
The first strategy involves relying on \emph{isolation} to prove that it is impossible for a given attacker to break some requirements. To formally capture this notion, we first extend the notion of IOC to attackers. The IOC of an attacker $A$, denoted $\blacktriangle(A)$, is defined by the union of IOCs of the components in $A$; more precisely, $\blacktriangle(A) \triangleq \bigcup\set{\blacktriangle(c)|c \in A}.$

Informally, isolation happens whenever the IOC of $A$ is disjoint from the COI of $r$, implying that $A$ cannot interact with $r$.
\begin{proposition}[Isolation]
\label{theo:isolation}
Let $A$ be an attacker and $r$ be a requirement that is satisfied in the absence of $A$. If $\blacktriangle(A)\cap \blacktriangledown(r)=\emptyset$, then $A$ cannot break $r$.\end{proposition}
\begin{proof}

For the attacker $A$ to break the requirement $r$, there must be a component $c \in \blacktriangledown(r)$ whose behaviour was affected by the presence of $A$, and whose change of behaviour caused $r$ to fail. However, for $A$ to affect the behaviour of $c$, there must be a dependency between the variables directly controlled by $A$ and $c$, since $A$ only chooses actions over the components it controls; implying that $c\in \blacktriangle(A)$. This contradicts the premise that the IOC of $A$ and the COI of $r$ are disjoint, so the component $c$ cannot exist. \qed
\end{proof}
Isolation reduces the SAT formula by dismissing attackers that are outside the COI of the requirement to be verified. Isolation works similarly to \emph{COI reduction} (see~\cite{ToSplitOrToGroup,GraphLabelingForEfficientCOIComputation,HandbookOfSatisfiability,HandbookOfModelChecking,OptimizedModelCheckingOfMultipleProperties}), and it transforms Equation \ref{eq:naiveCheck} into
\begin{align}
\label{eq:isolation}
\mathtt{check}(r,t)\triangleq\mathtt{goal}(r,t)\land\! \bigwedge_{c\in (\blacktriangledown(r)- W)}\left( \bigwedge_{k=0}^{t}{\mathtt{encode}(c,k)}\right)
\end{align}

The second strategy uses \emph{monotonicity} relation between capabilities and power of attackers.
\begin{proposition}[Monotonicity]
\label{theo:monotonicity}
For attackers $A$ and $B$ and a set of requirements $R$, if $A\subseteq B$, then $A[R]\subseteq B[R]$.
\end{proposition}
\begin{proof}
If $A$ is a subset of $B$, then attacker $B$ can always choose the same attack strategies that $A$ used to break the requirements in $A[R]$; thus, $A[R]$ must be a subset of $B[R]$.
 \qed
\end{proof}
Monotonicity allows us to define the notion of {minimal (successful) attackers} for a requirement $r$: attacker $A$ is a \emph{minimal attacker} for requirement $r$ if and only if $A$ breaks $r$, and there is no attacker $B\subset A$ such that $B$ also breaks $r$. In the remainder of this section, we expand on this notion, and we describe a methodology for attacker classification that focuses on the identification of these minimal attackers. 

\subsection{Minimal (Successful) Attackers}
The set of minimal attackers for a requirement $r$ partitions the set of attackers into those that break $r$ and those who do not. Any attacker that is more capable than a minimal attacker is guaranteed to break $r$ by monotonicity (cf. Proposition~\ref{theo:monotonicity}), and any attacker that is less capable than a minimal attacker cannot break $r$; otherwise, this less capable attacker would be a minimal attacker. Consequently, we can reduce the problem of attacker classification to the problem of finding the minimal attackers for all requirements.

\subsubsection{Existence of a Minimal Attacker.} Thanks to isolation (cf. Proposition~\ref{theo:isolation}) we can guarantee that a requirement $r$ that is safe in the absence of an attacker $A$ remains safe in the presence of $A$ if $\blacktriangledown(r) \cap \blacktriangle(A)$ is empty. Thus, for each requirement $r\in R$, the set of attackers that could break $r$ is $\ThePowersetOf{\blacktriangledown(r)-W}$. Out of all the attackers of $r$, the most capable attacker is $\blacktriangledown(r)-W$, so we can test whether there \emph{exists} any attacker that can break $r$ in $t$ steps by solving $\mathtt{check}(r,t)$ against attacker $\blacktriangledown(r)-W$. 
For succinctness, we henceforth denote the attacker $\blacktriangledown(r)-W$ by $r^{max}$.
\begin{corollary}
From monotonicity and isolation (cf. Propositions \ref{theo:monotonicity} and \ref{theo:isolation}), if attacker $r^{max}$ cannot break the requirement $r$, then there are no minimal attackers for $r$. Equivalently, if $r^{max}$ cannot break $r$, then $r$ does not belong to any set of broken requirements $A[R]$.
\end{corollary}
\begin{figure}[!h]
\centering
{
\begin{framed}
\begin{algorithm}[H]
 \KwData{system $S$, a requirement $r$, and a time step $t\geq 0$. }
\KwResult{set $M$ of \emph{minimal} attackers for $r$, bounded by $t$.}
{
\If{$check(r,t)$ is \textbf{not} satisfiable while assuming $c^\downarrow$ for all $c\in r^{max}$}
	{	
		\Return $\emptyset$\;
	}
}
Set: $P=\set{\emptyset}$, $M= \emptyset$; \quad /\!/($P$ contains the empty attacker $\emptyset$)\\
\While{$P$ is not empty}
{
	extract $A$ from $P$ such that the size of $A$ is minimal\;
	\If{\textbf{not} (exists $B \in M \text{ such that } B \subseteq A)$}
	{	
        	\eIf{$\mathtt{check}(r,t)$ is satisfiable when assuming $c^\downarrow$ for all $c\in A$}
        	{
        		insert $A$ in $M$\;
        	}
        	{
        		 \For{\!\!\textbf{each} $c \in (r^{max} - A)$}
        		 {
        			insert $A \cup \set{c}$ in $P$\;
        	 	}
        
        	}
	}
	
}
\Return $M$\;

 \caption{The $\mathtt{MinimalAttackers}$ algorithm.}
 \label{alg:CheckRequirement}
\end{algorithm}
\end{framed}}
\vspace{-0.5cm}
\end{figure}
\subsubsection{Finding Minimal Attackers.} 
After having confirmed that at least one minimal attacker for $r$ exists, we can focus on finding them. Our strategy consist of systematically increasing the capabilities of attackers that fail to break the requirement $r$ until they do. Algorithm~\ref{alg:CheckRequirement} describes this empowering procedure to computes the set of minimal attackers for a requirement $r$, which we call $\mathtt{MinimalAttackers}$. As mentioned, we first check to see if a minimal attacker exists (Lines 1-3); then we start evaluating attackers in an orderly fashion by always choosing the smallest attackers in the set of pending attackers $P$ (Lines 5-16). Line 7 uses monotonicity to discard the attacker $A$ if there is a successful attacker $B$ with $B\subseteq A$. Line 8 checks if the attacker $A$ can break $r$ in $t$ steps (or less); if so, then $A$ is a minimal attacker for $r$ and is included in $M$ (Line 9); otherwise, we empower $A$ with a new component $c$, and we add these new attackers to $P$ (Lines 11-13). We note that Line 11 relies on isolation, since we only add components that belong to the COI of $r$.

\begin{quote}
We recall the motivational example from Section~\ref{sec:Example}. Consider the computation of $\mathtt{MinimalAttackers}$ for requirement $r_2$. In this case, $r_2^{max}$ is $\set{g1,g2,v_1}$, which is able to break $r_2$, confirming the existence of (at least) a minimal attacker (Lines 1-3). We start to look for minimal attackers by checking the attacker $\emptyset$ (Lines 5-8); after we see that it fails to break $r_2$, we conclude that $\emptyset$ is not a minimal attacker and that we need to increase its capabilities. We then derive the attackers $\set{g_1}, \set{g_2}$ and $\set{v_1}$ by adding one non-isolated component to $\emptyset$, and we put them into the set of pending attackers (Lines 11-13). For attackers $\set{v_1}$ and $\set{g_2}$, we know that they can break the requirement $r_2$, so they get added to the set of minimal attackers, and are not empowered (Line 9); however, for attacker $\set{g_1}$, since it fails to break $r_2$, we increase its capabilities and we generate attackers $\set{v_1, g_1}$ and $\set{g_1, g_2}$. Finally, for these two latter attackers, since the minimal attackers $\set{v_1}$ and $\set{g_2}$ have already been identified, the check in Line 7 fails, and they are dismissed from the set of pending attackers, since they cannot be minimal. The algorithm finishes with $M=\set{\set{v_1}, \set{g_2}}$.
\end{quote}

Algorithm~\ref{alg:MinimalAttackers} applies Algorithm~\ref{alg:CheckRequirement} to each requirement; it collects all minimal attackers in the set $\mathcal{M}$ and initialises the attacker classification map $\mathcal{H}$. Finally, Algorithm~\ref{alg:GoodQuantification} exploits monotonicity to compute the classification of each attacker $A$ by aggregating the requirements broken by the minimal attackers that are subsets of $A$.

\begin{quote}
For the motivational example in Section~\ref{sec:Example}, Algorithm~\ref{alg:MinimalAttackers} returns $\mathcal{M}=\set{\emptyset,\set{v_1}, \set{g_2}}$ and $\mathcal{H}=\set{(\emptyset,\set{r_1}), (\set{v_1},\set{r_2, r_3}),  (\set{g_2},\set{r_2, r_3})}$. From there,  Algorithm~\ref{alg:GoodQuantification} completes the map $\mathcal{H}$, and returns
\begin{align*}
\mathcal{H}=\{&(\emptyset,\set{r_1}), (\set{v_1},\set{r_1,r_2, r_3}),  (\set{g_1},\set{r_1}),(\set{g_2},\set{r_1,r_2, r_3}), \\
&(\set{v_1,g_2},\set{r_1,r_2, r_3}),(\set{v_1,g_2},\set{r_1,r_2, r_3}),\\
&(\set{g_1,g_2},\set{r_1,r_2, r_3}),(\set{v_1,g_1,g_2},\set{r_1,r_2, r_3})\}
\end{align*}
\end{quote}
\vspace{-0.5cm}

\begin{figure}[!t]
\begin{framed}
\centering
{
\begin{algorithm}[H]
 \KwData{system $S$, a time step $t\geq 0$, and a set of requirements $R$.}
 \KwResult{Set of all minimal attackers $\mathcal{M}$ and an initial classification map $\mathcal{H}$.}
Set: $\mathcal{M}=\emptyset$\;
Map: $\mathcal{H}$\;
\For{\!\!\textbf{each} $r \in R$}
	{
		\For{\!\!\textbf{each} $A \in \mathtt{MinimalAttackers}(S,t,r)$}
		{
			insert $r$ in $\mathcal{H}(A)$\;
			insert $A$ in $\mathcal{M}$\;
		}
	 }
 \Return $(\mathcal{M},\mathcal{H})$\;
 \caption{The $\mathtt{AllMinimalAttackers}$ algorithm. }
 \label{alg:MinimalAttackers}
\end{algorithm}}
\end{framed}
\vspace{-0.5cm}
\end{figure}

\begin{figure}[!t]
\centering
{
\begin{framed}
\begin{algorithm}[H]
 \KwData{system $S=(W,V,G)$, a time step $t\geq 0$, a set of requirements $R$.}
 \KwResult{A map $\mathcal{H}$ that maps the attacker $A$ to $A[R]$.}

$(\mathcal{M},\mathcal{H})=\mathtt{AllMinimalAttackers}(S,t,R)$\;
\For{\!\!\textbf{each} $A \subseteq {(V\cup G)}$}
{
\For{\!\!\textbf{each} $A' \in \mathcal{M}$}
	{
		\If{$A' \subseteq A$}
		{
			insert all elements of $\mathcal{H}(A')$ in $\mathcal{H}(A)$\;
		}
	 }
}
 \Return $\mathcal{H}$\;
 \caption{Improved classification algorithm. We assume that $\mathcal{H}$ initially maps every $A$ to the empty set.}
 \label{alg:GoodQuantification}
\end{algorithm}
\end{framed}}
\vspace{-0.5cm}
\end{figure}

\subsection{On Soundness and Completeness}
\label{sec:completeness}
Just like any bounded model checking problem, if the time parameter $t$ is below the \emph{completeness threshold} (see \cite{EfficientComputationOfRecurrenceDiameters}), the resulting attacker classification up to $t$ steps could be \emph{incomplete}. More precisely, an attacker classification up to $t$ steps may prove that an attacker $A$ cannot break some requirement $r$ with a strategy up to $t$ steps, while in reality $A$ can break $r$ by using a strategy whose length is strictly greater than $t$. There are practical reasons that justify the use of a time parameter that is lower than the completeness threshold: 1) computing the exact completeness threshold is often as hard as solving the model-checking problem \cite{HandbookOfModelChecking}, so an approximation is taken instead; and 2), the complexity of the classification problem growths exponentially with $t$ in the worst case, since the size of the SAT formulae grow with $t$, and there is an exponential number of attackers that need to be classified by making calls to the SAT solver. A classification that uses a $t$ below the completeness threshold, while possibly incomplete, is \emph{sound}, i.e., it does not falsely report that an attacker can break a requirement when in reality it cannot. In Section~\ref{sec:discussion} we discuss possible alternatives to overcome incompleteness, but we leave a definite solution as future work.

We also consider the possibility of limiting the maximum size of minimal attackers to approximate the problem of attacker classification. The result of a classification whose minimal sets are limited to have up to $z$ elements is also sound but incomplete, since does not identify minimal attackers that have more than $z$ elements. We show in Section~\ref{sec:evaluation} that, even with restricted minimal attackers, it is possible to obtain a high coverage of the universe of attackers.
 
\subsection{Requirement Clustering}
\emph{Property clustering} \cite{ToSplitOrToGroup,HandbookOfSatisfiability,OptimizedModelCheckingOfMultipleProperties} is a state-of-the-art technique for the model checking of multiple properties. Clustering allows the SAT solver to reuse information when solving a similar instance of the same problem, but under different assumptions. To create clusters for attacker classification, we combine the SAT problems whose COI is similar (i.e., requirements that have a Jaccard index close to 1), and incrementally enable and disable properties during verification. More precisely, to use clustering, instead of computing $\mathtt{goal}(r,t)$ for a single requirement, we compute $\mathtt{goal}(Y,t)$ for a cluster $Y$ of requirements, defined by 
\begin{align}
\mathtt{goal}(Y,t) \triangleq \bigwedge_{r\in Y} (\lnot r^\downarrow \lor \mathtt{goal}(r,t)).
\end{align}
where $r^\downarrow$ is a new literal that plays a similar role to the ones used for gates and latches; i.e., we assume $r^\downarrow$ when we want to find the minimal attackers for $r$, and we assume $\lnot y^\downarrow$ for all other requirements $y \in Y$. 

The SAT problem for checking whether the cluster of requirements $Y$ is safe up to $t$ steps is
\begin{align}
\label{eq:MasterEq}
\mathtt{check}(Y,t)\triangleq\mathtt{goal}(Y,t)\land\! \bigwedge_{c\in (\blacktriangledown(Y)- W)}\left( \bigwedge_{k=0}^{t}{\mathtt{encode}(c,k)}\right),
\end{align}
where $\blacktriangledown(Y)=\bigcup \set{\blacktriangledown(r)| r\in Y}$.

\section{Evaluation}
\label{sec:evaluation}
In this section, we perform experiments to evaluate how effective is the use of isolation and monotonicity for the classification of attackers, and we evaluate the completeness of partial classifications for different time steps. 

For evaluation, we use a sample of AIG benchmarks from past Hardware Model-Checking Competitions (see \cite{HWMCC2011,HWMCC2013}), from their multiple-property verification track. Each benchmark has an associated list of invariants to be verified which, for the purposes of this evaluation, we interpret as the set of security requirements. As of 2014, the benchmark set was composed of 230 different instances, coming from both academia and industrial settings \cite{HWMCC2014BM}. We quote from \cite{HWMCC2014BM}:
\begin{quote}
``Among industrial entries, 145 instances
belong to the SixthSense family (6s*, provided by IBM), 24 are Intel benchmarks (intel*),
and 24 are Oski benchmarks. Among the academic related benchmarks, the set includes 13
instances provided by Robert (Bob) Brayton (bob*), 4 benchmarks coming from Politecnico
di Torino (pdt*) and 15 Beem (beem*). Additionally, 5 more circuits, already present in
previous competitions, complete the set.''
\end{quote}
All experiments are performed on a quad core MacBook with 2.9 GHz Intel Core i7 and 16GB RAM, and we use the SAT solver CaDiCaL version 1.0.3 \cite{Cadical}. The source code of the artefact is available at \cite{aig-ac-asset}.

We separate our evaluation in two parts: 1) a comparative study where we evaluate the effectiveness of using of monotonicity and isolation for attacker classification in several benchmarks, and 2) a case study, where we apply our classification methodology to a single benchmark --\texttt{pdtvsarmultip}-- and we study the results of varying the time parameters for partial classification.

\subsection{Evaluating Methodologies}
Given a set of competing classification methodologies $\mathcal{M}_1,\ldots, \mathcal{M}_n$ (e.g., Algorithm~\ref{alg:BadQuantification} and Algorithm~\ref{alg:GoodQuantification}), each methodology is given the same set of benchmarks $S_1,\ldots, S_m$, each with its respective set of requirements $R_1, \ldots, R_m$. To evaluate a methodology $\mathcal{M}$ on a benchmark $S=(W,V,G)$ with a set of requirements $R$, we allow $\mathcal{M}$ to ``learn'' for about 10 minutes per requirement by making calls to the SAT solver, and produce a (partial) attacker classification $\mathcal{H}$.  Afterwards, we compute the coverage metric obtained by $\mathcal{M}$, defined as follows.

\begin{definition}[Coverage]
Let $\ThePowersetOf{V}$ be the set of all attackers, and let $\mathcal{H}$ be the attacker classification produced by the methodology $\mathcal{M}$. We recall that $\mathcal{H}$ is a map that maps each attacker $A$ to a set of requirements, and in the ideal case, $\mathcal{H}(A)=A[R]$, for each attacker $A$. The \emph{attacker coverage obtained by methodology $\mathcal{M}$ for a requirement $r$} is the percentage of attackers $A\in\ThePowersetOf{V}$ for which we can correctly determine whether $A$ breaks $r$ by computing $r\in \mathcal{H}(A)$ (i.e., we do not allow guessing and we do not allow making new calls to the SAT solver).
\end{definition}

We also measure the execution time of the classification per requirement. More precisely, the time it takes for the methodology to find minimal attackers, capped at about 10 minutes per requirement. We force stop the classification for each requirement if a timeout occurs, but not while the SAT solver is running (i.e., we do not interrupt the SAT solver), which is why sometimes the reported time exceeds 10 minutes.

\subsection{Effectiveness of Isolation and Monotonicity}
To test the effectiveness of isolation and monotonicity, we selected a small sample of seven benchmarks. For each benchmark, we test four variations of our methodology: 
\begin{enumerate}
\item{$(+IS, +MO)$}: Algorithm~\ref{alg:GoodQuantification}, which uses both isolation and monotonicity
\item{$(+IS,-MO)$}: Algorithm~\ref{alg:GoodQuantification} but removing the check for monotonicity on Algorithm~\ref{alg:CheckRequirement}, Line 7;
\item{$(-IS,+MO)$}: Algorithm~\ref{alg:GoodQuantification} but using Equation~\ref{eq:isolation} instead of Equation~\ref{eq:MasterEq} to remove isolation while preserving monotonicity; and 
\item{$(+IS,+MO)$} Algorithm~\ref{alg:BadQuantification}, which does not use isolation nor monotonicity. 
\end{enumerate}

{The benchmarks we selected have an average of 173 inputs, 8306 gates, 517 latches, and 80 requirements. } Under our formulation of attackers, these benchmarks have on average $2^{8823}$ attackers
. However, since an attacker that controls a gate $g$ can be emulated by an attacker that controls all latches in the sources of $g$, we restrict attackers to be comprised of only latches; reducing the size of the set of attackers from $2^{8823}$ to $2^{517}$ on average per benchmark. Furthermore, we arbitrarily restrict the number of components that minimal attackers may control to a maximum of 3, which implies that, on a worst case scenario, we need to make a maximum of $80\times\sum_{k=0}^3 \binom {517}k$ calls to the SAT solver per benchmark. We also arbitrarily define the time step parameter $t$ to be 10.

Figure~\ref{fig:AverageCoverage} illustrates the average coverage for the four different methodologies, for each of the seven benchmarks. The exact coverage values are reported in Table~\ref{tab:appendix} in the Appendix. We see that our methodology consistently obtains the best coverage of all the other methodologies, with the exception of benchmark $\mathtt{6s155}$, where the methodology that removes isolation triumphs over ours. We attribute this exception to the way the SAT solver reuses knowledge when working incrementally; it seems that, for $(-IS,+MO)$, the SAT solver can reuse more knowledge than for $(+IS,+MO)$, which is why $(-IS,+MO)$ can discover more minimal attackers in average than $(+IS,+MO)$ (see Table~\ref{tab:appendix}).

We observe that the most significant element in play to obtain a high coverage is the use of monotonicity. Methodologies that use monotonicity always obtain better results than their counterparts without monotonicity. Isolation does not show a trend for increasing coverage, but has an impact in terms of classification time. Figure~\ref{fig:AverageExecTime} presents the average classification time per requirement for the benchmarks under the different methodologies. We note that removing isolation often increases the average classification time of classification methodologies; the only exception --benchmark $\mathtt{6s325}$-- reports a smaller time because the SAT solver ran out of memory during SAT solving about $50\%$ of the time, which caused an early termination of the classification procedure. This early termination also reflects on the comparatively low coverage for the method $(-IS,+MO)$ in this benchmark, reported on Figure~\ref{fig:AverageCoverage}.

\begin{figure}[!t]
\centering
\includegraphics[width=\textwidth]{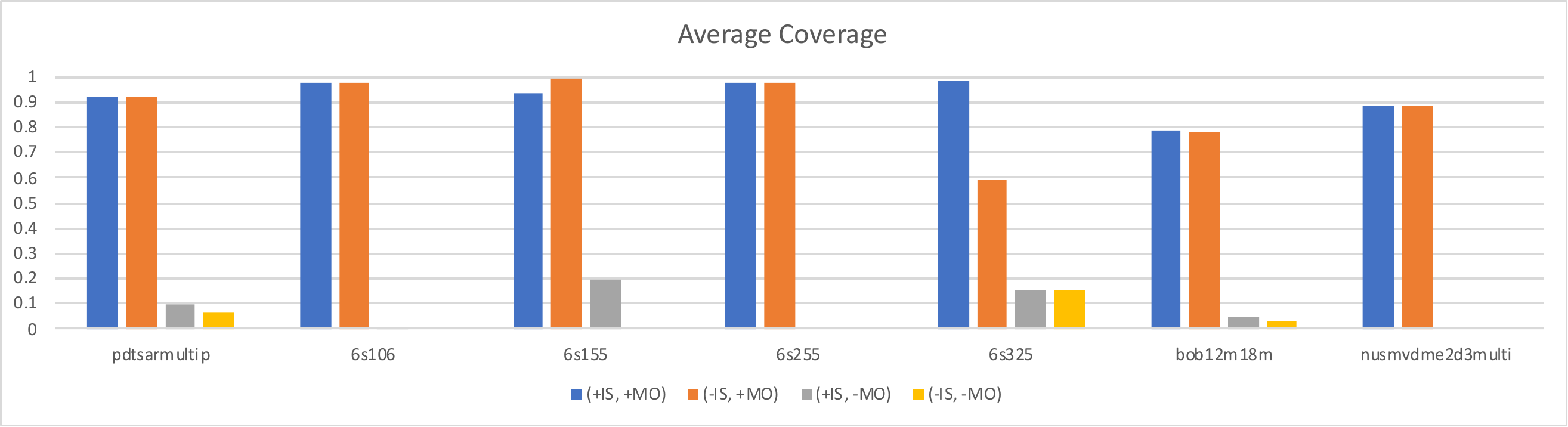}
\caption{Average requirement coverage per benchmark. A missing bar indicates a value that is approximately 0.}
\label{fig:AverageCoverage}
\vspace{0.5cm}
\includegraphics[width=\textwidth]{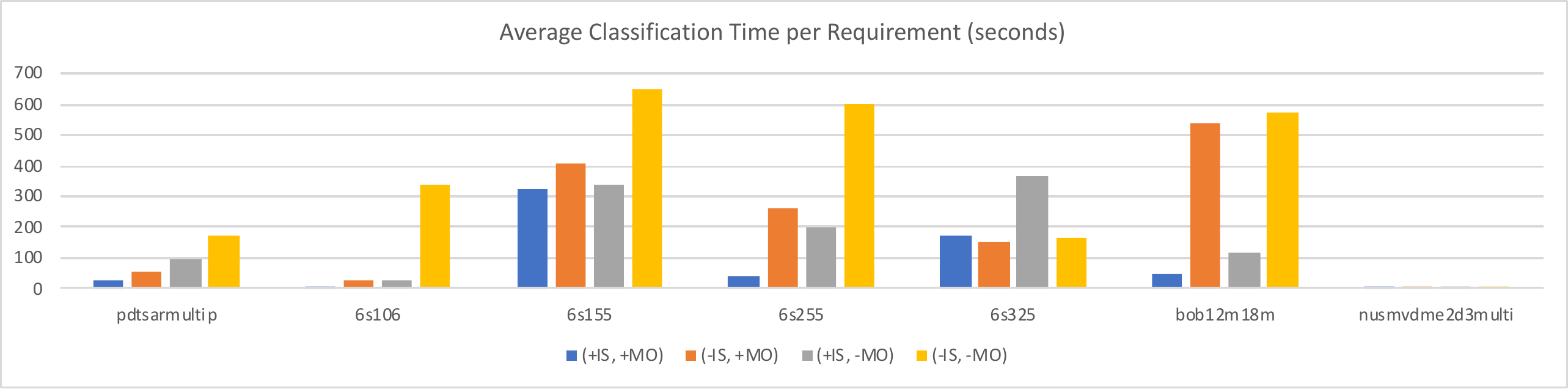}
\caption{Average classification time per requirement per benchmark. The time for benchmark nusmvdme2d3multi is very close to 0 in all instances. }
\label{fig:AverageExecTime}
\end{figure}

\subsection{Partial Classification of the \texttt{pdtvsarmultip} Benchmark}
The benchmark \texttt{pdtvsarmultip} has 17 inputs, 
130 latches, 
2743 gates, 
and has an associated list of 33 invariant properties, out of which 31 are unique and we interpret as the list of security requirements. 
Since we are only considering attackers that control latches, there are a total of $2^{130}$ attackers that need to be classified for the 31 security requirements.

We consider 6 scenarios for partial classification up to $t$, with $t$ taking values in $\set{0,1,5,10,20,30}$. For each requirement, we obtain the execution time of classification (ms), the size of the set of source latches for the requirement (\#C), the number of minimal attackers found (\#Min.), the total number of calls to the SAT solver (\#SAT), the average number of components per minimal attacker (\#C./Min) and the coverage for the requirement (Cov.). We present the average of these measures in Table~\ref{tab:pdt}. 

Normally, the attacker classification behaves in a similar way to what is reported for requirement $ \Always \lnot g_{2177}$, shown in Table~\ref{tab:pdt2324}. More precisely, coverage steadily increases and stabilises as we increase $t$. However, we like to highlight two interesting phenomena that may occur: 1) coverage {may} \emph{decrease} as we increase the time step (e.g., as shown in Table~\ref{tab:pdt2367}
), and 2) the number of minimal attackers decreases while the coverage increases, as shown in Tables~\ref{tab:pdt2324} and ~\ref{tab:pdt2367}.

Case 1) occurs because the set of attackers that can effectively interact with the system at time 0 is rather small, i.e., $2^6$, while the set of attackers that can affect the system at times 0 and 1 has size $2^{26}$. The size of this set increases with time until it stabilises at $2^{66}$, which is the size of the set of attackers that cannot be dismissed by isolation. 

Case 2) occurs because the minimal attackers that are found for smaller time steps represent a small percentage of the set of attackers that can affect the system, so there is very little we can learn by using monotonicity. More precisely, those minimal attackers control a relatively large set of components, which they need to be successful in breaking requirements, as shown in Step 5, column \#C./Min in Tables~\ref{tab:pdt2324} and ~\ref{tab:pdt2367}. By considering more time steps, we are allowing attackers that control less components to further propagate their actions through the system, which enables attack strategies that were unsuccessful for smaller choices of time steps.

\begin{table}[!t]
\centering
\begin{tabular}{|c|c|c|c|c|c|c|}
\hline
Steps & ms &  \#C. & \#Min. & \#SAT & \#C./Min. & Cov.\\
\hline
0 & 683.1290323 & 34.96774194 & 2.451612903 & 16328.77419 & 1.4 & 0.527277594\\
1 & 2387.548387 & 46.22580645 & 6.387096774 & 24420 & 1.650232484 & 0.572533254\\
5 & 5229.935484 & 58.93548387 & 44.90322581 & 28355.29032 & 1.639645689 & 0.84956949\\
10 & 24967.12903 & 58.93548387 & 151.1935484 & 25566.54839 & 1.460869285 & 0.918973269\\
20 & 13632.51613 & 58.93548387 & 17.67741935 & 20849.70968 & 1.176272506 & 0.979354259\\
30 & 12208.25806 & 58.93548387 & 15.93548387 & 20798.16129 & 1.104563895 & 0.979354274\\
\hline
\end{tabular}
\caption{Average measures for all requirements per time steps.}
\label{tab:pdt}
\centering
\vspace{0.5cm}
\begin{minipage}{0.475\textwidth}
\centering
\resizebox{\textwidth}{!}{
\begin{tabular}{|c|c|c|c|c|c|c|}
\hline
\multicolumn{7}{|c|}{$\Always \lnot g_{2177}$} \\
\hline
Steps & ms & \#C & \#Min. & \#SAT & \#C./Min. & Cov.\\
\hline
0 & 895 & 59 & 0 & 34281 & -- & 5.94E-14\\
1 & 2187 & 66 & 10 & 47378 & 2 & 0.499511\\
5 & 1735 & 66 & 205 & 12476 & 1.912195 & 0.999997\\
10 & 968 & 66 & 27 & 9948 & 1 & 0.999999\\
20 & 1275 & 66 & 27 & 9948 & 1 & 0.999999\\
30 & 1819 & 66 & 27 & 9948 & 1 & 0.999999\\
 \hline
\end{tabular}
}
\caption{Coverage for requirement $\Always \lnot g_{2177}$.}
\label{tab:pdt2324}
\end{minipage}
\begin{minipage}{0.475\textwidth}
\centering
\resizebox{\textwidth}{!}{
\begin{tabular}{|c|c|c|c|c|c|c|}
\hline
\multicolumn{7}{|c|}{$\Always \lnot g_{2220}$} \\
\hline
Steps & ms &  \#C. & \#Min. & \#SAT & \#C./Min. & Cov.\\
\hline
0 & 1 & 6 & 1 & 28 & 1 & 0.90625\\
1 & 86 & 26 & 1 & 2628 & 1 & 0.500039\\
5 & 4511 & 67 & 17 & 47664 & 2.588235 & 0.852539\\
10 & 3226 & 67 & 6 & 37889 & 1 & 0.984375\\
20 & 3355 & 67 & 6 & 37889 & 1 & 0.984375\\
30 & 3562 & 67 & 6 & 37889 & 1 & 0.984375\\
 \hline
\end{tabular}
}
\caption{Coverage for requirement $\Always \lnot g_{2220}$.}
\label{tab:pdt2367}
\end{minipage}
\end{table}
By taking an average over all requirements, we observe that coverage seems to steadily increase as we increase the number of steps for the classification, as reported in Table~\ref{tab:pdt}, column Cov. The low coverage for small $t$ is due to the restriction on the size of minimal attackers. More precisely, for small $t$, attackers can only use short strategies, which limits their interaction with the system; we expect attackers to control a large number of components if they want to successfully influence a requirement in this single time step, and since we restricted our search to attackers of size 3 maximum, these larger minimal attackers are not found (e.g., as reported in Table~\ref{tab:pdt2324} for Step 0).

We conclude that experimental evidence favours the use of both monotonicity and isolation for the classification of attackers, although  some exceptions may occur for the use of isolation. Nevertheless, these two techniques help our classification methodology $(+IS,+MO)$ consistently obtain significantly better coverage when compared to the naive methodology $(-IS,-MO)$.

\section{Related Work}
\label{sec:discussion}

\subsubsection{On Defining Attackers.} Describing an adequate attacker model to contextualise the security guarantees of a system is not a trivial task. Some attacker models may be adequate to provide guarantees over one property (e.g. confidentiality), but not for a different one (e.g., integrity). Additionally, depending on the nature of the system and the security properties being studied, it is sensible to describe attackers at different levels of abstraction. For instance, in the case of security protocols, Basin and Cremers define attackers in~\cite{KnowYourEnemy} as combinations of compromise rules that span over three dimensions:  \emph{whose} data is compromised, \emph{which} kind of data it is, and \emph{when} the compromise occurs. In the case of Cyber-physical Systems (CPS), works like \cite{Giraldo2018,Simei} model attackers as sets of components (e.g., some sensors or  actuators), while other works like \cite{IFCPSSec,Cardenas2011,Urbina2016} model attackers that can arbitrarily manipulate any control inputs and any sensor measurements at will, as long as they avoid detection. In the same context of CPS, Rocchetto and Tippenhauer \cite{CPSProfiles} model attackers more abstractly as combinations of quantifiable traits (e.g., insider knowledge, access to tools, and financial support), which, when provided a compatible system model, ideally fully define how the attacker can interact with the system. 

Our methodology for the definition of attackers combines aspects from~\cite{KnowYourEnemy,Giraldo2018} and \cite{Simei}. The authors of~\cite{KnowYourEnemy} define symbolic attackers and a set of rules that describe how the attackers affect the system, which is sensible since many cryptographic protocols are described symbolically. Our methodology describes attackers as sets of components (staying closer to the definitions of attackers in \cite{Giraldo2018} and \cite{Simei}), and has a lower level of abstraction since we describe the semantics of attacker actions in terms of how they change the functional behaviour of the AIG, and not in terms of what they ultimately represent. This lower level of abstraction lets us systematically and exhaustively generate attackers by simply having a benchmark description, but it limits the results of the analysis to the benchmark; Basin and Cremers can compare among different protocol implementations, because attackers have the same semantics even amongst different protocols. If we had an abstraction function from sets of gates and latches to symbolic notions (e.g., ``components in charge of encryption'', or ``components in charge of redundancy''), then it could be possible to compare results amongst different AIGs.

\subsubsection{On Efficient Classification.} 
The works by Cabodi, Camurati and Quer \cite{GraphLabelingForEfficientCOIComputation}, Cabodi et. al \cite{ToSplitOrToGroup}, and Cabodi and Nocco \cite{OptimizedModelCheckingOfMultipleProperties} present several useful techniques that can be used to improve the performance of model checking when verifying multiple properties, including COI reduction and property clustering. We also mention the work by Goldberg et al. \cite{JustAssume} where they consider the problem of efficiently checking a set of safety properties $P_1$ to $P_k$ by individually checking each property while assuming that all other properties are valid. Ultimately, all these works inspired us to incrementally check requirements in the same cluster, helping us transform Equation~\ref{eq:naiveCheck} into Equation~\ref{eq:MasterEq}. Nevertheless, we note that all these techniques are described for model checking systems in the absence of attackers, which is why we needed to introduce the notions of isolation and monotonicity to account for them. Finally, it may be possible to use or incorporate other techniques that improve the efficiency of BMC in general (e.g., interpolation \cite{Interpolation}).

\subsubsection{On Completeness.} As mentioned in Section~\ref{sec:completeness}, if the time parameter for the classification is below the {completeness threshold}, the resulting attacker classification is most likely {incomplete}. To guarantee completeness, it may be possible to adapt existing termination methods (see \cite{ProvingMorePropertiesWithBMC}) to consider attackers. Alternatively, methods that compute a good approximation of the completeness threshold (see \cite{EfficientComputationOfRecurrenceDiameters}) which guarantee the precision of resulting the coverage should help improve the completeness of attacker classifications. Interpolation \cite{Interpolation} can also help finding a guarantee of completeness. Also, the verification techniques IC3 \cite{IC3,IC32} and PDR \cite{IC3}, which have seen some success in hardware model checking, may address the limitation of boundedness.

\subsubsection{On Verifying Non-Safety Properties.} In this work, we focused our analysis exclusively on safety properties of the form $\Always e$. However, we believe that it is possible to extend this methodology to other types of properties, it is possible to efficiently encode Linear Temporal Logic formulae for bounded model checking \cite{BMCWithoutBDDs,lmcs:2236}. The formulations of the SAT problem change for the different nature of the formulae, but both isolation and monotonicity should remain valid heuristics, since they ultimately refer to how strategies of attackers can be inferred, not how they are constructed.  

\section{Conclusion and Future Work}
\label{sec:conclusion}
In this work, we present a methodology to model check systems in the presence of attackers with the objective of mapping each attacker to the list of security requirements that it breaks. This mapping of attackers creates a classification for them, defining equivalence classes of attackers by the set of requirements that they can break. The system can then be considered safe in the presence of attackers that cannot break any requirement. While it is possible to perform a classification of attackers by exhaustively performing model checking, the exponential size of the set of attackers renders this naive approach impractical. Thus, we rely on ordering relations between attackers to efficiently classify a large percentage of them, and we demonstrate empirically by applying our methodology to a set of benchmarks that describe hardware systems at a bit level.

In our view, ensuring the completeness of the attacker classification is the most relevant direction for future work. Unlike complete classifications, incomplete classifications cannot provide guarantees that work in the general case if minimal attackers are not found. We also note that the effectiveness of monotonicity for classification is directly related to finding minimal attackers. Consequently, our methodology may benefit from any other method that helps in the identification of those minimal attackers. In particular, we are interested in checking the effectiveness of an approach where, instead of empowering attackers, we try to \emph{reduce} successful attackers into minimal attackers by removing unnecessary capabilities. This, formally, is an \emph{actual causality analysis} \cite{ActualCausality} of successful attackers.

\bibliographystyle{splncs03}

\appendix
\section*{Appendix}
\begin{landscape}
\begin{table}[!t]
\centering
\centering
\begin{tabular}{|c|c|c|c|c|c|c|c|}
\hline
Benchmark & pdtsarmultip & 6s106 & 6s155 & 6s255 & 6s325 & bob12m18m & nusmvdme2d3multi\\
\hline
\multicolumn{8}{|c|}{Average Coverage per Requirement}\\
\hline
(+IS, +MO) & 0.918973269 & 0.977620187 & 0.93762207 & 0.977172852 & 0.986949581 & 0.785075777 & 0.8853302\\
(-IS, +MO) & 0.916622423 & 0.977008259 & 0.99609375 & 0.9765625 & 0.590205544 & 0.781704492 & 0.8853302\\
(+IS, -MO) & 9.95E-02 & 2.01E-03 & 0.1953125 & 1.73828E-11 &  0.156147674 & 4.61E-02 & 4.52E-15\\
(-IS, -MO) & 6.45E-02 & 9.10E-36 & 6.00E-72 & 2.0838E-222 & 0.156146179 & 3.29E-02 & 4.52E-15\\
\hline
\multicolumn{8}{|c|}{Average Classification Time per Requirement}\\
\hline
(+IS, +MO) & 24967.12903 & 3088.882353 & 326372.2188 & 38396.125 & 174066.6047 & 49916.59211 & 3530.666667\\
(-IS, +MO) & 56187.48387 & 25090.47059 & 408165.4063 & 263457.125 & 152425.3889 & 541870.25 & 3348.333333\\
(+IS, -MO) & 97326.29032 & 28483.11765 & 340715.0313 & 197275.75 & 363913.8472 & 114330.0855 & 3890\\
(-IS, -MO) & 170630.7742 & 341430.7059 & 646926.0625 & 602012.5 & 162959.5436 & 572080.0987 & 4262.333333\\
\hline
\multicolumn{8}{|c|}{Average Number of Components to Build Attackers From per Requirement}\\
\hline
(+IS, +MO) & 58.93548387 & 35.47058824 & 9 & 93.875 & 173.1196013 & 124.0460526 & 63\\
(-IS, +MO) & 82 & 135 & 257 & 752 & 1668 & 261 & 63\\
(+IS, -MO) & 58.93548387 & 35.47058824 & 9 & 93.875 &  173.1196013 & 124.0460526 & 63\\
(-IS, -MO) & 82 & 135 & 257 & 752 & 1668 & 261 & 63\\
\hline
\multicolumn{8}{|c|}{Average Number of Identified Successful Attackers per Requirement}\\
\hline
(+IS, +MO) & 151.1935484 & 16.05882353 & 7.5 & 2.5625 & 237.3056478 & 34.93421053 & 129.3333333\\
(-IS, +MO) & 146.0322581 & 16.05882353 & 8 & 2.4375 & 46.42424242 & 34.68421053 & 129.3333333\\
(+IS, -MO) & 15860.25806 & 14693.17647 & 98 & 13792.9375 &  214293.01 & 22700.58553 & 1282\\
(-IS, -MO) & 27207.54839 & 105842.7059 & 126133.2813 & 20657 & 2502.25641 & 40526.72368 & 1282\\
\hline
\multicolumn{8}{|c|}{Average Number of Calls to SAT Solver per Requirement}\\
\hline
(+IS, +MO) & 25566.54839 & 1978.647059 & 10.4375 & 11382.125 & 583123.7209 & 300825.9605 & 40576.33333\\
(-IS, +MO) & 58420.83871 & 294477.9412 & 2022640.25 & 35252.1875 & 51411.24747 & 1607803.138 & 40576.33333\\
(+IS, -MO) & 41275.6129 & 16655.76471 & 101 & 23455.875 & 550652.3721 & 322031.1513 & 41729\\
(-IS, -MO) & 85390.80645 & 396510 & 1389339.563 & 49366.1875 & 56238.01026 & 1266308.526 & 41729\\
\hline
\multicolumn{8}{|c|}{Average Minimal Number of Components Needed by Successful Attackers per Requirement}\\
\hline
(+IS, +MO) & 1.460869285 & 1.004524887 & 1 & 0.464285714 & 1.649279888 & 1.142237928 & 2.843533741\\
(-IS, +MO) & 1.460869285 & 1.004524887 & 1 & 0.4375 & 1 & 1.123402209 & 2.843533741\\
(+IS, -MO) & 2.912532493 & 2.909934933 & 2.1953125 & 2.885495873 & 2.927835684 & 2.964760733 & 2.984195449\\
(-IS, -MO) & 2.975171828 & 2.983304749 & 2.665659086 & 1.986925862 & 1.977314997 & 2.973301799 & 2.984195449\\
\hline
\end{tabular}
\caption{Average metrics per requirement for the different benchmarks.}
\label{tab:appendix}
\end{table}
\end{landscape}

\end{document}